\documentclass[10pt]{amsart}

\usepackage{amsmath}

\DeclareMathOperator*{\supp}{supp}
\DeclareMathOperator*{\dv}{div}

\usepackage{mathbbol}
\usepackage{tikz}

\newcommand{\D}{\Delta}
\newcommand{\Dd}{\Delta_\ddagger}
\newcommand{\na}{\nabla}
\newcommand{\nad}{\nabla_\ddagger}
\newcommand{\nab}{\nabla_\bullet}
\newcommand{\bx}{\mathbf{x}}
\newcommand{\bv}{\mathbf{v}}
\newcommand{\bw}{\mathbf{w}}

\newcommand{\bc}{\mathbf{c}}
\newcommand{\be}{\mathbf{e}}

\newcommand{\pd}{\partial}
\newcommand{\om}{\omega}

\newcommand{\Omh}{\Omega^h}

\newcommand{\Sih}{\Sigma^h}
\newcommand{\dif}{\mathrm{d}}
\newcommand{\bnt}{\mathbf{0}}
\newcommand{\bbv}{\pmb{\mathbb{v}}^h}
\newcommand{\bbp}{\mathbb{p}^h}

\newcommand{\Fhb}{\mathbf{F}^h}

\newcommand{\tvt}{{\tilde{v}}^h_3}
\newcommand{\tv}{{\tilde{\bv}}^h}

\newcommand{\vrem}{{\bv}^h_{rem}}
\newcommand{\tw}{\tilde{\bw}}
\newcommand{\prem}{p^h_{rem}}
\newcommand{\premb}{\overline{p}^h_{rem}}
\newtheorem{thm}{Theorem}
\newtheorem{Rem}{Remark}
\newtheorem{prop}{Proposition}
\newtheorem{lem}{Lemma}
\newtheorem{cor}{Corollary}
\newtheorem{cla}{Claim}

\begin{document}

\title{Modified Reynolds equation for steady flow through a curved pipe}

\author{A. Ghosh}
 \address{Mathematics and Applied Mathematics, MAI, Link\"oping University, SE 58183 Link\"oping, Sweden}
 \email{arpan.ghosh@liu.se}

\author{V. A. Kozlov}
 \address{Mathematics and Applied Mathematics, MAI, Link\"oping University, SE 58183 Link\"oping, Sweden}
 \email{vladimir.kozlov@liu.se}

\author{S. A. Nazarov}%\thanks{S. A. Nazarov acknowledges the support from Russian Foundation of Basic Research, grant 15-01-02175.}
 \address{St. Petersburg State University, 198504, Universitetsky pr., 28, Stary Peterhof, Russia; Institute of Problems of Mechanical Engineering RAS, V.O., Bolshoj pr., 61, St. Petersburg, 199178, Russia}
 \email{srgnazarov@yahoo.co.uk}

\begin{abstract}
	A modified Reynolds equation governing the steady flow of a fluid with low Reynolds number through a curvilinear, narrow tube, with its derivation from Stokes equations through asymptotic methods is presented. The channel considered may have large curvature and torsion. Approximations of the velocity and the pressure of the fluid inside the channel are constructed by artificially imposing appropriate boundary conditions at the inlet and the outlet. A justification for the approximations is provided along with a comparison with a simpler case.
\end{abstract}

\maketitle

\section{Introduction}

Fluid flow through narrow tubes have a wide range of applications including industrial, vehicular and biological systems. This has generated interest in the mathematical study of such flows and various results have been produced based on different cases, configurations and assumptions, see for example \cite{PaPi15II,CaMi} and references therein. A vast majority of the existing works are about straight pipes. Straight narrow pipes with varying cross-sections are considered in \cite{NaPi90}. Curvilinear pipes with torsion and curvature have been studied in \cite{Maru01}, although a fixed cross-section is assumed throughout the length of the pipe. However, in \cite{Maru01}, the effects of the curvature and torsion can only be seen in the second term of the  asymptotic approximation as the leading terms are free of these effects. The Reynolds equation is widely used to describe flow through thin channels and its derivation from the Stokes equations can be found in \cite{BaCh86,Cim}. 

We consider steady flows through thin pipes having a variable geometry of the cross-section while allowing large longitudinal curvature and torsion. We aim to construct an analogue of the Reynolds equation that can be successfully used for this general situation. Certainly, there are some natural restrictions, for example, on the longitudinal curvature.

Our main task is to derive a one-dimensional Reynolds equation for the pressure and an equation for the longitudinal component of the velocity resulting in a modified Poiseuille flow. To construct these terms, we need two boundary conditions which are usually prescribed through pressure and/or flux. The discrepancy is represented by the boundary layer terms that are significant near the end cross-sections. This is the reason that we choose artificial boundary conditions which allow us to easily estimate the discrepancy.

%The new equation can then lead to asymptotic approximations such that the leading terms are enough to reflect the effects of the non-trivial geometry thereby eliminating the need for solving extra equations for higher order terms, thus enabling easier implementation to practice.

\subsection{Nondimensionalization of the problem}
Let us consider a curvilinear pipe of length $L$ with varying channel diameter and non-circular cross-section. Let us denote the mean radius by $H$ which satisfies $H=Lh$. We assume that the pipe is narrow or in other words, $h\ll1$ is a small parameter. Let us denote the steady state velocity and the kinematic pressure of the flowing fluid by $\textbf{V}$ and $P$ respectively. These satisfy the Navier-Stokes system
\begin{gather*}
-\nu\D_\textbf{X}\textbf{V}+(\textbf{V}\cdot\na_\textbf{X})\textbf{V}+\na_\textbf{X} P=0,\\
-\na_\textbf{X}\cdot\textbf{V}=0,
\end{gather*}
complemented by suitable boundary conditions and where $\nu$ is the kinematic viscosity of the fluid flowing through the tube and $\textbf{X}$ signifies the position vector in suitable units of length.

Let $F$ be the flux through a cross-section. We introduce new dimensionless variables for our problem:
\[
\bx=\frac{1}{L}\textbf{X},\quad \bv=\frac{H^2}{F}\textbf{V}\quad\mbox{and}\quad p=\frac{H^2L}{\nu F}P.
\]

In terms of the new dimensionless quantities and the Reynolds number $\textsf{Re}:=\frac{FL}{\nu H^2}$, the Stokes system takes the form
\[
-\D_\bx\bv+\textsf{Re}(\bv\cdot\na_\bx)\bv+\na_\bx p=0.
\]
For the article, we assume a small Reynolds number so that we are able to discard the convective term.

\subsection{Formulation of the problem}
Let $\bc$ denote the arc-length parameterized centre curve for the pipe in consideration with $s\in[0,1]$ being the arc-length parameter. Let $hR>0$ give the distance of the interior boundary of the pipe from $\bc(s)$ along a direction which is perpendicular to the centre curve at $s$. 

Let us consider the regions
\begin{align*}
		\Omh&=\{\bx(r,\theta,s):\eta=h^{-1}r,0\leq r< hR(\theta,s),\theta\in[0,2\pi),s\in(0,1)\},\\
		\Sih&=\{\bx\in\pd\Omh:s\in(0,1)\},\\
		\om(s)&=\{\bc(s)+\eta\be_1(\theta,s):\eta	=h^{-1}r,0\leq r< hR(\theta,s),\theta\in[0,2\pi)\}.
\end{align*}
The region $\Omh$ is assumed to be locally Lipschitz and each transversal cross-section, hence $\om(s)$ as well, to be a star domain for every $s$..

The goal is to find an asymptotic approximation of the solution of the Stokes problem 
\begin{align}\label{Stokes}
	-\D_\bx\bv^h+\na_\bx p^h&=\bnt\quad\mbox{in}\quad\Omh,\\\label{Div}
				-\na_\bx\cdot\bv^h&=	 0\quad\mbox{in}\quad\Omh,\\\label{NoSlip}
										 \bv^h&=\bnt\quad\mbox{on}\quad\Sih,
\end{align}
supplemented with appropriate boundary conditions at the cross-sections $s=0$ and $s=1$.

\subsection{Results}
The primary highlight of this article is the derivation of a modified Reynolds equation
\[-\pd_s(G(s)\pd_sp^0(s))=0,\quad s\in(0,1),
\]
for flow through a curved pipe. Here, the function $G$, defined as \[G(s):=2\!\!\int\limits_{\om(s)}\!\!{\Psi(\eta,\theta,s)\eta\dif\eta\dif\theta},\] depends on the geometry of the pipe and $p^0$ is the leading term in the formal asymptotic expansion of $p^h$ as stated in \eqref{Ans}. Defining $\beta(\eta,\theta,s)$ as the scale factor corresponding to the longitudinal parameter $s$ and $\nad$ as the gradient operator on $\om(s)$, the function $\Psi$ is obtained as the solution of \[-\nad\cdot\beta\nad\Psi=2\quad\mbox{in}\quad\om(s),\quad\Psi=0\quad\mbox{on}\quad\pd\om(s).
\] 

The modified Reynolds equation takes into consideration a relatively wide range of curvature as well as variation of the diameter of the pipe as characterized by the following restrictions:
\begin{align}\label{asmp}
	|\bc'''|\leq ch^{-2+2\delta},\quad|\bc''''|\leq ch^{-3+3\delta},\quad|\pd_sR|\leq ch^{-1+2\delta}\quad\mbox{and}\quad|\pd_s^2R|\leq ch^{-2+3\delta}
\end{align}
for some positive $\delta$. In particular, the first inequality implies \[|\bc''|\leq ch^{-1+\delta}.\] Our approach in this article is to treat these geometrical quantities as separate parameters in the beginning and choosing the optimal orders by the end of the analysis. The new equation covers the case of smooth curvature (order $0$ w.r.t h) and nearly constant radius (order $1$ w.r.t h) of the pipe as well which is achieved by replacing the right hand sides in \eqref{asmp} by a constant independent of $h$.

%Another important contribution of our article is in the Lemmas \ref{Lem1} and \ref{Lem2}, where we present estimates for the divergence equation specifically for curved, thin, tubular domains. For the case of thin tubular domains, in the previous works starting with \cite{NaPi90}, coordinate dilation and  uniform scaling of the transversal velocity components were sufficient to derive the specific estimates. The presence of curvature complicates our case, therefore, position dependent scaling is introduced to tackle this problem.

Based on the modified Reynolds equation, under the said assumptions and provided the appropriate boundary conditions, we proceed to construct an approximation $\{\bbv,\bbp\}$ for the solution $\{\bv^h,p^h\}$ of \eqref{Stokes}-\eqref{NoSlip} in the form
\begin{equation*}%\label{71}
\begin{aligned}
					\bbp(\bx)&=h^{-3}p^0(s),\\
\bbv(\bx)&=h^{-1}v^1_3(\eta,\theta,s)\bc'(s)+\bv^2_\ddagger(\eta,\theta,s).\\
 %\bbv_\ddagger(\bx)&=X^h(s)\bv^2_\ddagger(\eta,\theta,s),
\end{aligned}
\end{equation*}
Accordingly, we obtain the representation 
\begin{equation}
\bv^h=\bbv+\vrem\quad\mbox{and}\quad p^h=\bbp+\prem.
\label{Sol}
\end{equation}

We finally prove that the error terms $\{\vrem,\prem\}$ in \eqref{Sol} admit the bounds
\begin{align}\label{rest}
	h\|\na_\bx\vrem\|+\|\vrem\|+h^2\|\prem-\premb\|\leq ch^{\delta}
\end{align}
whereas $\{\bbv,\bbp\}$ are estimated by
\begin{align}\label{appest}
	h\|\na_\bx\bbv\|+\|\bbv\|+h^2\|\bbp-\overline{\mathbb{p}}^h\|\leq ch^{0}
\end{align}
thereby justifying the approximate solution for any positive $\delta$.

\section{Geometry and notations}
\begin{figure}%
\begin{tikzpicture}
	\draw[rotate=45] (0,5) ellipse (.5 and 1);
	%\draw (0,5) ellipse (.5 and 1);
	\draw[rotate=45] (5,0) ellipse (1 and .5);
	\draw[dashed][rotate=45] (0,5) to [out=0,in=90] (5,0);
	\draw[rotate=45] (0,6) to [out=0,in=90] (6,0);
	\draw[rotate=45] (0,4) to [out=0,in=90] (4,0);
	\node at (2,4) {$\Omh$};
	\draw[->] (0,5) to (0,5.75);
	\draw[->] (0,5) to (0.75,5);
	\draw[->] (0,5) to (-0.4,4.6);
	\node[right] at (0,5.75) {$\be_1(\theta,s)$};
	\node[right] at (-0.4,4.6) {$\be_2(\theta,s)$};
	\node[right] at (0.75,5) {$\bc'(s)$};
	\node[above left] at (0,5) {$\bc(s)$};
	\draw[<->] (-2.5,4.33) to (-3,5.2);
	\node[right] at (-2.85,5) {$hR$};
\end{tikzpicture}
\caption{The curvilinear coordinate frame $\{\be_1,\be_2,\bc'\}$ depicted at the point $\bc(s)$ on the centre curve in the domain $\Omh.$}%
\label{fig1}%
\end{figure}
The ambient three dimensional space is taken to have a canonical Cartesian coordinate system. The initial direction of the curve is assumed to be along the third coordinate direction, i.e., $\bc'(0)=(0,0,1)^T.$ The vector quantities for the problem are described using the coordinate frame consisting of the triplet $\{\be_1(\theta,s),\be_2(\theta,s),\bc'(s)\}$, depicted in Figure \ref{fig1}, where $\be_i$ are obtained by solving
\[
\pd_s\be_i(\theta,s)=-(\bc''(s)\cdot\be_i(\theta,s))\,\bc'(s)
\]
with the initial conditions
\[\be_1(\theta,0)=(\cos{\theta},\sin{\theta},0)^T\mbox{ and }\be_2(\theta,0)=(-\sin{\theta},\cos{\theta},0)^T.\]
Additionally, they also satisfy
\[\pd_\theta\be_1(\theta,s)=\be_2(\theta,s)\mbox{ and }\pd_\theta\be_2(\theta,s)=-\be_1(\theta,s).\]
\begin{Rem}
	The frame $\{\be_1(\theta,s),\be_2(\theta,s),\bc'(s)\}$ is an orthonormal frame of reference. As opposed to the Frenet-Serret frame, it is well defined even in the curvature-free segments of the pipe.
\end{Rem}
The parameter $\theta\in[0,2\pi)$ signifies the direction from a point on $\bc$, along the plane perpendicular to $\bc'$ at that point, with respect to some reference direction. Throughout this article, vectors are denoted in bold while their components along $\be_1,\be_2$ and $\bc'$ are given by the corresponding letter with subscripts $1,2$ and $3$ respectively.

With this frame, we have new curvilinear coordinates $\{r,\theta,s\}$ which are related to the Cartesian coordinates by
\[
\bx(r,\theta,s)=\bc(s)+r\be_1(\theta,s),\quad 0\leq r\leq hR(\theta,s).
\]
Note that in the absence of curvature, $\{r,\theta,s\}$ are cylindrical coordinates for the tube.

Clearly, $R$ must be positive and sufficiently smooth and we define
\begin{equation}\label{Hbnd1}
	\gamma:=\max\limits_{(\theta,s)\in[0,2\pi)\times[0,1]}|\pd_sR(\theta,s)|
\end{equation}
and
\begin{equation}\label{Hbnd2}
	\gamma^*:=\max\limits_{(\theta,s)\in[0,2\pi)\times[0,1]}|\pd_s^2R(\theta,s)|.
\end{equation}

Additionally, we also define 
\begin{align}\label{3der}
	\lambda:=\max\limits_{s\in[0,1]}|h\bc'''(s)|
\end{align}
along with
\begin{align}\label{4der}
	\lambda^*:=\max\limits_{s\in[0,1]}|h\bc''''(s)|.
\end{align}
Also as a result,
\begin{align}\label{2der}
	|\bc''(s)|=\sqrt{|\bc'''(s)\cdot\bc'(s)|}\leq h^{-1/2}\lambda^{1/2} \quad\forall s\in[0,1].
\end{align}

Let $\nab$ be the two dimensional gradient operator on a cross-section, i.e., \[\nab:=\be_1\pd_r+\frac{1}{r}\be_2\pd_\theta.\]
Let $\nad$ and $\Dd$ denote the components of the gradient operator and the Laplacian respectively, on a cross-section in terms of the scaled parameters $\eta:=h^{-1}r$ and $\theta$, i.e., \[\nad:=\be_1\pd_\eta+\frac{1}{\eta}\be_2\pd_\theta=h\nab\mbox{ and }\Dd:=\nad\cdot\nad=\pd_\eta^2+\frac{1}{\eta}\pd_\eta+\frac{1}{\eta^2}\pd_\theta^2.\]
Then, we have
\begin{align*}
	\na_\bx=\nab+\beta^{-1}\bc'\pd_s=h^{-1}\nad+\beta^{-1}\bc'\pd_s,\\
\end{align*}
and
\begin{gather*}
\D_\bx=h^{-2}\Dd-h^{-1}\beta^{-1}\bc''\!\cdot\!\nad+\beta^{-2}\pd_s^2+h\beta^{-3}\eta\bc'''\!\cdot\!\be_1\pd_s\\
=h^{-2}\beta^{-1}\nad\cdot\beta\nad+\beta^{-1}\pd_s\beta^{-1}\pd_s
\end{gather*}
where we have introduced the scale factor \[\beta(r,\theta,s):=|\pd_s\bx(r,\theta,s)|=1-r\bc''(s)\cdot\be_1(\theta,s)=1-h\eta\bc''(s)\cdot\be_1(\theta,s)\] corresponding to the parameter $s$ and used the fact that $\nab\beta=h^{-1}\nad\beta=-\bc''.$ 
Then due to \eqref{3der}, \eqref{4der} and \eqref{2der}, we have
\begin{equation}
|\pd_s\beta|\leq c\lambda\quad\mbox{and}\quad|\pd_s^2\beta|\leq c(\lambda^*+h^{-1/2}\lambda^{3/2}).
\label{2derb}
\end{equation}

The physical restriction on the curvature such that $h\bc''(s)\cdot\be_1(\theta,s)<R(\theta,s)^{-1}$ for all $(\theta,s)\in[0,2\pi)\times[0,1]$ eliminates the possibility of the pipe curving into itself. Moreover, to ensure the validity of the asymptotic procedure followed in this article, we must additionally assume
\begin{equation}
\lambda=o(h^{-1}).
\label{hl}
\end{equation}

For integration over the cross sections, we have the area element 
\[\dif\sigma(r,\theta)=r\dif r\dif\theta.\]
On the other hand, the volume element is given in the new coordinates as
\[\dif\bx(r,\theta,s)=\beta(r,\theta,s)\dif\sigma(r,\theta)\dif s.\]

Let us denote the velocity vector component wise as  $\bv^h=v^h_1\be_1+v^h_2\be_2+v^h_3\bc'$.
The components of the quantities in \eqref{Stokes} along any cross-section $\om(s)$ of the pipe satisfy
\begin{equation}
\begin{aligned}
	-h^{-2}\beta^{-1}\nad\cdot\beta\nad\bv^h_\ddagger+h^{-1}\nad p^h+\beta^{-2}\bc''\bc''\cdot\bv^h_\ddagger-\bc'''_\ddagger\beta^{-2}v^h_3-\bc''\beta^{-2}\pd_sv^h_3\\
	-\bc''\beta^{-1}\pd_s(\beta^{-1}v^h_3)-\sum\limits_{i=1,2}\be_i\beta^{-1}\pd_s(\beta^{-1}\pd_sv^h_i)=\bnt_\ddagger\quad\mbox{in}\quad\Omh.
\end{aligned}
\label{Stokesdd}
\end{equation}
On the other hand, \eqref{Stokes} results in the following equation for the direction along the length of the pipe:
\begin{equation}
\begin{aligned}
	-h^{-2}\beta^{-1}\nad\cdot\beta\nad v^h_3+\beta^{-2}\bc'''\cdot\bv^h_\ddagger+\sum\limits_{i=1,2}\bc''\cdot\be_i\beta^{-1}\pd_s(\beta^{-1}v^h_i)\\
	-\beta^{-2}\bc'''\cdot\bc'v^h_3-\beta^{-1}\pd_s(\beta^{-1}\pd_sv^h_3)+\beta^{-1}\pd_sp^h=0\quad\mbox{in}\quad\Omh.
\end{aligned}
\label{Stokes3}
\end{equation}
Finally, the divergence equation \eqref{Div} can be reformulated as
\begin{equation}
	-h^{-1}\beta^{-1}\nad\cdot\beta\bv^h_\ddagger-\beta^{-1}\pd_sv^h_3=0\quad\mbox{in}\quad\Omh.
\label{Div0}
\end{equation}
In the above and henceforth, $\ddagger$ in the subscript of the vector symbols denote their respective projections onto the cross-sectional plane.

In order to ensure uniqueness of the asymptotic solution of the above problem, we intend to impose additional artificial conditions at the ends of the pipe. We shall argue in the next sections that a prescribed flux at the inlet and an ambient (possibly atmospheric) pressure condition at the outlet are sufficient for our purpose.
%\[
%\int\limits_{\om(0)}\!\!v_3^h(r,\theta,0)\dif\sigma(\eta,\theta)=F^hh\quad\mbox{and}\quad p^h(r,\theta,1)=p^h_{per}.
%\]

\section{Model problems and estimates}

In this section, we present the estimates related to some model problems that we rely upon in the asymptotic procedure. We use similar notations for function spaces as in \cite{Temam}, which include standard notations for Sobolev spaces. In particular, the function spaces denoted by bold letters represent the corresponding space of vector/tensor valued functions of the appropriate dimension.

\subsection{Stokes system}

We first consider a modified Stokes problem on the two-dimensional domain $\om(s)$. We present the relevant estimates in the theorem that follows.
\begin{thm}\label{thm1}
Let there be given $\mathbf{f}\in \mathbf{H}^{-1}(\om(s))$, $g\in L^2(\om(s))$ and $\mathbf{h}\in \mathbf{H}^{1/2}(\pd\om(s))$ satisfying the compatibility condition,
\begin{equation}
\int\limits_{\om(s)}{\beta g}\dif\sigma(\eta,\theta)+\int\limits_{0}^{2\pi}{\beta\mathbf{h}\cdot(R\be_1-(\pd_\theta R)\be_2)\dif\theta}=0.
\label{compat}
\end{equation}
Then there exist a unique $\mathbf{u}\in \mathbf{H}^{1}(\om(s))$ and a unique $q\in L^2(\om(s))$ up to a constant that solve the two-dimensional modified Stokes problem
\begin{equation}\label{modeq1a}
\begin{aligned}
	-\beta^{-1}\nad\cdot\beta\nad\mathbf{u}+\nad q=\mathbf{f},\quad-\beta^{-1}\nad\cdot\beta\mathbf{u}=g\quad\mbox{in}\quad\om(s),\\
	\mathbf{u}=\mathbf{h}\quad\mbox{on}\quad\pd\om(s).
\end{aligned}
\end{equation}
The solutions admit the estimate
\begin{equation}\label{mod1a}
\begin{aligned}
	\|\nad\mathbf{u}\|_{\mathbf{L}^2(\om(s))}+\|\mathbf{u}\|_{\mathbf{L}^2(\om(s))}+\|q-\bar{q}\|_{L^2(\om(s))}\\\leq c(\|\mathbf{f}\|_{\mathbf{H}^{-1}(\om(s))}+\|g\|_{L^2(\om(s))}+\|\mathbf{h}\|_{\mathbf{H}^{1/2}(\pd\om(s))}),
	\end{aligned}
\end{equation}
where $\bar{q}$ is $q$ averaged over $\om(s)$.

Furthermore, if $\mathbf{f}\in \mathbf{L}^2(\om(s))$, $g\in H^{1}(\om(s))$ and $\mathbf{h}\in \mathbf{H}^{3/2}(\pd\om(s))$, then $\mathbf{u}\in \mathbf{H}^{2}(\om(s))$ and $q\in H^{1}(\om(s))$ satisfy
\begin{equation}\label{mod1b}
\begin{aligned}
	\|\nad\nad\mathbf{u}\|_{\mathbf{L}^2(\om(s))}+\|\nad q\|_{\mathbf{L}^2(\om(s))}\\\leq c(\|\mathbf{f}\|_{\mathbf{L}^2(\om(s))}+\|g\|_{H^{1}(\om(s))}+\|\mathbf{h}\|_{\mathbf{H}^{3/2}(\pd\om(s))}).
	\end{aligned}
\end{equation}
\end{thm}
\begin{proof}
We accept \eqref{mod1a} without proof as it is a standard estimate for generalised  Stokes systems, see e.g. \cite{HuSt}. that can be applied to this case owing to the boundedness of the parameter $\beta$. In order to obtain \eqref{mod1b}, we rewrite \eqref{modeq1a} as
\begin{equation*}
\begin{aligned}
	-\nad\cdot\nad\mathbf{u}+\nad q=\mathbf{f}+h\beta^{-1}\bc''\cdot\nad\mathbf{u},\quad-\nad\cdot\mathbf{u}=g+h\beta^{-1}\bc''\cdot\mathbf{u}\quad\mbox{in}\quad\om(s),\\
	\mathbf{u}=\mathbf{h}\quad\mbox{on}\quad\pd\om(s).
\end{aligned}
\end{equation*}
Using the boundedness of $\beta$ and \eqref{2der}, we have the following estimate due to results in \cite{Temam}.
\[
\begin{aligned}
	\|\nad\nad\mathbf{u}\|_{\mathbf{L}^2(\om(s))}+\|\nad q\|_{\mathbf{L}^2(\om(s))}\leq c(\|\mathbf{f}\|_{\mathbf{L}^2(\om(s))}+h^{1/2}\lambda^{1/2}\|\nad\mathbf{u}\|_{\mathbf{L}^2(\om(s))}\\+\|g\|_{H^{1}(\om(s))}+h^{1/2}\lambda^{1/2}\|\mathbf{u}\|_{\mathbf{H}^1(\om(s))}+\|\mathbf{h}\|_{\mathbf{H}^{3/2}(\pd\om(s))}).
	\end{aligned}
\]
Then applying \eqref{mod1a}, we get \eqref{mod1b} by using \eqref{hl}.
\end{proof}
We have the following corollary as a consequence of the above theorem.
\begin{cor}\label{cor1}
Given $\mathbf{f}\in\mathcal{C}^1((0,1),\mathbf{L}^2(\om(s)))$, $g\in\mathcal{C}^1((0,1),H^{1}(\om(s)))$ and $\mathbf{h}\in\mathcal{C}^1((0,1),\mathbf{H}^{3/2}(\pd\om(s)))$ such that \eqref{compat} holds for every $s\in(0,1)$, then the solution of \eqref{modeq1a} satisfies the estimate
\begin{equation}\label{mod1c}
\begin{aligned}
	\|(\pd_s\mathbf{u})_\ddagger\|_{\mathbf{H}^1(\om(s))}+\|\pd_s(q-\bar{q})\|_{L^2(\om(s))}\leq c(\|(\pd_s\mathbf{f})_\ddagger\|_{\mathbf{H}^{-1}(\om(s))}+\|\pd_sg\|_{L^2(\om(s))}\\+\|\pd_s\mathbf{h}\|_{\mathbf{H}^{1/2}(\pd\om(s))}+(\lambda+\gamma)(\|\mathbf{f}\|_{\mathbf{L}^2(\om(s))}+\|g\|_{H^1(\om(s))}+\|\mathbf{h}\|_{\mathbf{H}^{3/2}(\pd\om(s))}).
	\end{aligned}
\end{equation}
\end{cor}
\begin{proof}
Differentiating \eqref{mod1a} with respect to $s$, we get the system of equations
\begin{align*}
	-\beta^{-1}\nad\!\cdot\!\beta\nad(\pd_s\mathbf{u})_\ddagger+\nad (\pd_sq)=(\pd_s\mathbf{f})_\ddagger+\beta^{-1}\nad\!\cdot\!(\pd_s\beta)\nad\mathbf{u}-\beta^{-2}(\pd_s\beta)\nad\!\cdot\!\beta\nad\mathbf{u},\\
	-\beta^{-1}\nad\cdot\beta(\pd_s\mathbf{u})_\ddagger=\pd_sg+\beta^{-1}\nad\!\cdot\!(\pd_s\beta)\mathbf{u}-\beta^{-2}(\pd_s\beta)\nad\!\cdot\!\beta\mathbf{u}\quad\mbox{in}\quad\om(s),\\
	\pd_s\mathbf{u}=\pd_s\mathbf{h}-(\pd_sR)\pd_\eta\mathbf{u}\quad\mbox{on}\quad\pd\om(s).
\end{align*}
If the condition \eqref{compat} corresponding to the above system is satisfied, then we can apply Theorem \ref{thm1}.
\begin{cla}\label{Claim}
For every $s\in(0,1)$,
\begin{align*}
	\int\limits_{\om(s)}{\beta(\pd_sg+\beta^{-1}\nad\!\cdot\!(\pd_s\beta)\mathbf{u}-\beta^{-2}(\pd_s\beta)\nad\!\cdot\!\beta\mathbf{u})}\dif\sigma(\eta,\theta)\\+\int\limits_{0}^{2\pi}{\beta(\pd_s\mathbf{h}-(\pd_sR)\pd_\eta\mathbf{u})\cdot(R\be_1-(\pd_\theta R)\be_2)\dif\theta}=0.
\end{align*}
\end{cla}
The proof is presented in the appendix.
Applying Theorem \ref{thm1}, we find
\begin{gather*}
	\|(\pd_s\mathbf{u})_\ddagger\|_{\mathbf{H}^1(\om(s))}+\|\pd_s(q-\bar{q})\|_{L^2(\om(s))}\leq c(\|(\pd_s\mathbf{f})_\ddagger\|_{\mathbf{H}^{-1}(\om(s))}\\+\|\pd_sg\|_{L^2(\om(s))}+\|\pd_s\mathbf{h}\|_{\mathbf{H}^{1/2}(\pd\om(s))}+\lambda\|\nad\mathbf{u}\|_{\mathbf{L}^2(\om(s))}+\gamma\|\pd_\eta\mathbf{u}\|_{\mathbf{H}^{1/2}(\pd\om(s))})
\end{gather*}
where we have used \eqref{Hbnd1} and \eqref{3der}. Then we estimate $\|\nad\mathbf{u}\|_{\mathbf{L}^2(\om(s))}$ using \eqref{mod1a} and $\|\pd_\eta\mathbf{u}\|_{\mathbf{H}^{1/2}(\pd\om(s))}$ using \eqref{mod1b} to get \eqref{mod1c}.
\end{proof}

\subsection{The elliptic system}
The next theorem provides us the estimates for the model problem for scalar functions that appear in the asymptotic procedure. The results are standard (see e.g. \cite{LiMa72}) and hence the proof is omitted.
\begin{thm}\label{thm2}
Let there be given $f\in H^{-1}(\om(s))$ and $k\in H^{1/2}(\pd\om(s))$. Then there exists a unique $u\in H^{1}(\om(s))$ solving
\begin{equation}\label{modeq2a}
\begin{aligned}
	-\beta^{-1}\nad\cdot\beta\nad u=f\quad&\mbox{in}\quad\om(s),\\
u=k\quad\mbox{on}\quad\pd\om(s).
\end{aligned}
\end{equation}
The solution admits the following estimate:
\begin{align}\label{mod2a}
	\|u\|_{H^1(\om(s))}\leq c(\|f\|_{H^{-1}(\om(s))}+\|k\|_{H^{1/2}(\pd\om(s))}).
\end{align}
In general, if $f\in H^{n-2}(\om(s))$ and $k\in H^{n-1/2}(\pd\om(s))$ for $n\geq1$, then $u\in H^{n}(\om(s))$ satisfies
\begin{equation}\label{mod2b}
\begin{aligned}
	\|u\|_{H^n(\om(s))}\leq c(\|f\|_{H^{n-2}(\om(s))}+\|k\|_{H^{n-1/2}(\pd\om(s))}).
	\end{aligned}
\end{equation}
\end{thm}
Consequently, we have the following corollary.
\begin{cor}\label{cor2}
Given $f\in\mathcal{C}^2((0,1),H^{n}(\om(s)))$, $k\in\mathcal{C}^2((0,1),H^{n+3/2}(\pd\om(s))),$ the solution of \eqref{modeq2a} satisfies the estimates
\begin{equation}\label{mod2c}
\begin{aligned}
	\|\pd_su\|_{H^n(\om(s))}\leq c(\|\pd_sf\|_{H^{n-2}(\om(s))}+\|\pd_sk\|_{H^{n-1/2}(\pd\om(s))}\\+(\lambda+\gamma)(\|f\|_{H^{n-1}(\om(s))}+\|k\|_{H^{n+1/2}(\pd\om(s))}))
	\end{aligned}
\end{equation}
and
\begin{equation}
\begin{aligned}
	\|\pd_s^2u\|_{H^n(\om(s))}\leq c(\|\pd_s^2f\|_{H^{n-2}(\om(s))}+\|\pd_s^2k\|_{H^{n-1/2}(\pd\om(s))}\\+(\lambda+\gamma)(\|\pd_sf\|_{H^{n-1}(\om(s))}+\|\pd_sk\|_{H^{n+1/2}(\pd\om(s))})\\+(\lambda^*+\gamma^*+h^{-1/2}\lambda^{3/2}+\gamma^2)(\|f\|_{H^{n}(\om(s))}+\|k\|_{H^{n+3/2}(\pd\om(s))})).
	\end{aligned}
\label{mod2d}
\end{equation}
\end{cor}
\begin{proof}
To prove \eqref{mod2c}, we follow identical steps as in the proof of Corollary \ref{cor1}. To prove \eqref{mod2d}, we derive \eqref{modeq2a} twice with respect to $s$ to obtain
\begin{align*}
	-\beta^{-1}\nad\!\cdot\!\beta\nad\pd_s^2u=\pd_s^2f+2\beta^{-1}\nad\!\cdot\!(\pd_s\beta)\nad\pd_su-2\beta^{-2}(\pd_s\beta)\nad\!\cdot\!\beta\nad\pd_su\\
	-2\beta^{-2}(\pd_s\beta)\nad\!\cdot\!(\pd_s\beta)\nad u+\beta^{-1}\nad\!\cdot\!(\pd_s^2\beta)\nad u-(\pd_s(\beta^{-2}\pd_s\beta))\nad\!\cdot\!\beta\nad u\quad\mbox{in}\quad\om(s),\\
	\pd_s^2u=\pd_s^2k-2(\pd_sR)\pd_\eta\pd_su-(\pd_s^2R)\pd_\eta u-(\pd_sR)^2\pd_\eta^2\mathbf{u}\quad\mbox{on}\quad\pd\om(s).
\end{align*}
Then we apply Theorem \ref{thm2} and use \eqref{3der}, \eqref{2derb}, \eqref{Hbnd1} and \eqref{Hbnd2} to get
\begin{align*}
	\|\pd_s^2u\|_{H^1(\om(s))}\leq c(\|\pd_s^2f\|_{H^{-1}(\om(s))}+(\lambda^2+\lambda^*+h^{-1/2}\lambda^{3/2})\|\nad u\|_{L^2(\om(s))}\\
	+\lambda\|\nad\pd_su\|_{L^2(\om(s))}+\|\pd_s^2k\|_{H^{1/2}(\pd\om(s))}+\gamma\|\pd_\eta\pd_s u\|_{H^{1/2}(\pd\om(s))}\\+\gamma^*\|\pd_\eta  u\|_{H^{1/2}(\pd\om(s))}+\gamma^2\|\pd_\eta^2u\|_{H^{1/2}(\pd\om(s))}).
\end{align*}
Estimating the right hand side with the help of \eqref{mod2a}, \eqref{mod2b} and \eqref{mod2c} and using \eqref{hl}, we arrive at \eqref{mod2d}.
\end{proof}

\subsection{The divergence equation}

In this subsection, we consider the divergence equation for two different cases of a curvilinear pipe having a variable cross-section. The divergence equation frequently appears in the study of flows and hence is an important auxiliary problem, see \cite{Galdi,Lady}. For the case of thin tubular domains, in the previous works starting with \cite{Naz90}, coordinate dilation and  uniform scaling of the transversal velocity components were sufficient to derive the specific estimates. See also \cite{Papi15}. The presence of curvature complicates our case, therefore, position dependent scaling involving the curvature dependent scale factor $\beta$ is introduced to tackle this problem.

Firstly, we present a Lemma about the divergence equation in a thin curvilinear pipe $\Omh$ laving length $1$.
\begin{lem}\label{Lem1}
Let there be $f\in L^2(\Omh)$ such that
\begin{equation}
\int\limits_{\Omh}{f}\dif\bx=0.
\label{perp}
\end{equation}
Then there exists a (non unique) solution $\bw\in \mathbf{H}^{1}(\Omh)$ of the divergence equation
\begin{equation}
\begin{aligned}
	-\na_\bx\cdot\bw&=f\mbox{ in }   \Omh,\\
							 \bw&=\bnt\mbox{ on }\pd\Omh,
\end{aligned}
\label{nabla}
\end{equation}
which obeys the estimate
\begin{equation}
\begin{aligned}
	\|\nab\bw_\ddagger\|_{L^2(\Omh)}+h^{-1}\|\bw_\ddagger\|_{L^2(\Omh)}+	h\|\nab w_3\|_{L^2(\Omh)}\\+\|\pd_sw_3\|_{L^2(\Omh)}+\|w_3\|_{L^2(\Omh)}\leq C\|f\|_{L^2(\Omh)},\\
	h^{-1}\|(\pd_s\bw_\ddagger)_\ddagger\|_{L^2(\Omh)}\leq C(1+\lambda)\|f\|_{L^2(\Omh)}
\end{aligned}
\label{lem1}
\end{equation}
for some constant $C$ independent of $f$ and $h$.
\end{lem}
\begin{proof}
Noting the fact that $\nad\beta=-h\bc''$ and in accordance with the scaled parameter $\eta=h^{-1}r,$ we introduce the scaled function
\[\hat{\bw}=h^{-1}\beta\bw_\ddagger+w_3\bc'.\]
Thus, we have
\begin{align*}
	\na_\bx\cdot\bw&=h^{-1}\nad\cdot\bw_\ddagger+\beta^{-1}(\pd_sw_3-\bc''\cdot\bw_\ddagger)=\beta^{-1}(\nad\cdot\hat{\bw}_\ddagger+\pd_s\hat{w}_3)\\
								 &=\beta^{-1}(\pd_\eta\hat{w}_1+\eta^{-1}\hat{w}_1+\eta^{-1}\pd_\theta\hat{w}_2+\pd_s\hat{w}_3).
\end{align*}
Clearly, the terms within the brackets in the last equality represent the polar form of the divergence of a vector field defined in a straight cylinder. As a result, we can say that $\bw$ satisfies \eqref{nabla} if and only if $\bar{\bw}:=\hat{w}_1\hat{\boldsymbol{\eta}}+\hat{w}_2\hat{\boldsymbol{\theta}}+\hat{w}_3\hat{\mathbf{s}}$ (likewise for  $\hat{\boldsymbol{\eta}},\hat{\boldsymbol{\theta}}$ and $\hat{\mathbf{s}}$ being the unit vectors corresponding to cylindrical coordinates) satisfies the system
\begin{align*}
	\dv\bar{\bw}&=\beta f\mbox{ in }\Xi,\\
		 \bar{\bw}&=\bnt.
\end{align*}
Here $\Xi$ is a cylinder with a straight axis and given as 
\[\Xi:=\{\bx(\eta,\theta,s)=(\eta\cos\theta,\eta\sin\theta,s):0\leq\eta\leq h^{-1}R(\theta,s),0\leq\theta<2\pi,0<s<1\}.\]
For the function $\beta f,$ we have that
\begin{gather*}
	\int\limits_{\Xi}\!\!{\beta f}\dif\bx
	=\int\limits_0^1\int\limits_0^{2\pi}\int\limits_0^{h^{-1}R(\theta,s)}\!\!\beta f\eta\dif\eta\dif\theta\dif s
	=\int\limits_{\Omh}\!\!{f}\dif\bx=0.
\end{gather*}
Thus the compatibility condition is met by $\beta f.$

Therefore, by a classical result on the divergence equation (see \cite{LaSo76}) in a fixed Lipschitz domain, we have $\bar{\bw}\in\mathbf{H}^1(\Xi)\Rightarrow\hat{\bw}\in\mathbf{H}^1(\Omh)$ and for a constant $C$ independent of the data, the estimate
\begin{gather*}
	\|\nad\hat{\bw}_\ddagger\|_{\mathbf{L}^2(\Omh)}+\|(\pd_s\hat{\bw}_\ddagger)_\ddagger\|_{\mathbf{L}^2(\Omh)}+\|\hat{\bw}_\ddagger\|_{\mathbf{L}^2(\Omh)}\\+	\|\nad\hat{w}_3(f)\|_{\mathbf{L}^2(\Omh)}+\|\pd_s\hat{w}_3(f)\|_{L^2(\Omh)}+\|\hat{w}_3(f)\|_{L^2(\Omh)}
	\leq C\|f\|_{L^2(\Omh)}.
\end{gather*}
Owing to the bounds \eqref{3der} and \eqref{2der}, the above leads us to \eqref{lem1}.
\end{proof}

We present another lemma on the divergence equation restricted to a length of a pipe that is comparable to the thickness of the pipe. The estimate in this case is modified as compared to that in Lemma \ref{Lem1} due to the differing aspect ratio of the segment of the curvilinear pipe in question. Let us consider \eqref{nabla} and \eqref{perp} restricted to the domain $\Omh_{end}\!:=\!\{\bx(r,\theta,s)\!\in\!\Omh\!:\!0\!<\!s\!<\!l,l=O(h)\}.$
\begin{lem}\label{Lem2}
Let $f\in L^2(\Omh_{end})$ satisfy
\begin{equation}
\int\limits_{\Omh_{end}}{f}\dif\bx=0.
\label{perpend}
\end{equation}
Then there exists a (non unique) solution $\bw\in \mathbf{H}^{1}(\Omh_{end})$ of the divergence equation
\begin{equation}
\begin{aligned}
	-\na_\bx\cdot\bw&=f\mbox{ in }   \Omh_{end},\\
							 \bw&=\bnt\mbox{ on }\pd\Omh_{end},
\end{aligned}
\label{nablaend}
\end{equation}
which obeys the estimate
\begin{equation}
\begin{aligned}
	\|\nab\bw_\ddagger\|_{\mathbf{L}^2(\Omh_{end})}+\|(\pd_s\bw_\ddagger)_\ddagger\|_{\mathbf{L}^2(\Omh_{end})}+h^{-1}\|\bw_\ddagger\|_{\mathbf{L}^2(\Omh_{end})}\\+	\|\nab w_3\|_{\mathbf{L}^2(\Omh_{end})}+\|\pd_sw_3\|_{L^2(\Omh_{end})}+h^{-1}\|w_3\|_{L^2(\Omh_{end})}\leq C\|f\|_{L^2(\Omh_{end})},
\end{aligned}
\label{lem2}
\end{equation}
for some constant $C$ independent of $f$ and $h.$
\end{lem}
\begin{proof}
We introduce the scaled parameters $\eta=h^{-1}r$ and $\tau=h^{-1}s$ and the scaled function
\[\hat{\bw}=\beta\bw_\ddagger+w_3\bc'.\]
The rest of the proof follows the steps in the proof of Lemma \ref{Lem1} and we get the required estimate.
\end{proof}

\section{Formal asymptotic procedure}

Let us consider the asymptotic Ans\"atze:
\begin{equation}
\begin{aligned}
	p^h(r,\theta,s)&=h^{-3}	 p^0(s)+h^{-2}p^1(\eta,\theta,s)+h^{-1}p^2(\eta,\theta,s)+\ldots,\\
\bv^h(r,\theta,s)&=h^{-1}\bv^1(\eta,\theta,s)+h^0 \bv^2(\eta,\theta,s)+\ldots.
\end{aligned}
\label{Ans}
\end{equation}
Having an $O(h^{-1})$ velocity still results in an $O(h)$ flux through the cross-sections so that ignoring the convective term in the Navier-Stokes equations can still be justified.  

The first step of matching coefficients of the leading order of $h$ in \eqref{Stokesdd}, \eqref{Div0} and \eqref{NoSlip} produces the following system of equations:
\begin{align*}
-\beta^{-1}\nad\cdot\beta\nad\bv^1_\ddagger+\nad p^1=\bnt_\ddagger,\quad-\beta^{-1}\nad\cdot\beta\bv^1_\ddagger=0\quad\mbox{in}\quad\om(s),\\
\bv^1_\ddagger=\bnt_\ddagger\quad\mbox{on}\quad\pd\om(s).
\end{align*}
The solution is of the form $\bv^1_\ddagger=\bnt_\ddagger$ and $p^1=p^1(s).$

For the third component, due to \eqref{Stokes3} and \eqref{NoSlip}, we have the equations
\begin{equation}
\begin{aligned}
-\beta^{-1}\nad\cdot\beta\nad v^1_3+\beta^{-1}\pd_sp^0=0\quad&\mbox{in}\quad\om(s),\\
v^1_3=0\quad\mbox{on}\quad\pd\om(s).
\end{aligned}
\label{55}
\end{equation}
Hence, we have as a solution
\begin{equation}
v^1_3=-\frac{1}{2}\Psi(\eta,\theta,s)\pd_sp^0(s),
\label{Pois}
\end{equation}
where $\Psi$ is a function (Prandtl function in case of $\beta\equiv 1$) satisfying
\begin{equation}
-\beta^{-1}\nad\cdot\beta\nad\Psi=2\beta^{-1}\quad\mbox{in}\quad\om(s),\quad\Psi=0\quad\mbox{on}\quad\pd\om(s).
\label{Pran}
\end{equation}
For the solution of \eqref{Pran}, due to \eqref{mod2a}, we have the estimate
\begin{equation}
\|\Psi\|_{H^1(\om(s))}\leq c\|\beta^{-1}\|_{H^{-1}(\om(s))}\leq c.
\label{psi}
\end{equation}
Applying Corollary \ref{cor2} and using \eqref{2der} and \eqref{2derb}, we obtain the additional estimates
\begin{align}\label{dpsi}
	\|\pd_s\Psi\|_{L^2(\om(s))}&\leq c[\|\pd_s(\beta^{-1})\|_{H^{-1}(\om(s))}+(\lambda+\gamma)\|\beta^{-1}\|_{L^{2}(\om(s))}]\leq c(\lambda+\gamma).\\\label{ddpsi}
\|\pd_s^2\Psi\|_{L^2(\om(s))}&\leq c[\|\pd_s^2(\beta^{-1})\|_{H^{-1}(\om(s))}+(\lambda+\gamma)\|\pd_s(\beta^{-1})\|_{L^{2}(\om(s))}\\\nonumber
														 &+(\lambda^*+\gamma^*+h^{-1/2}\lambda^{3/2}+\gamma^2)\|\beta^{-1}\|_{H^{1}(\om(s))}]\\\nonumber
														 &\leq c(\lambda^*+\gamma^*+h^{-1/2}\lambda^{3/2}+\gamma^2).
\end{align}

We also need the boundedness of the functions $\Psi$ and $\Psi^{-1}$ to proceed further and hence we present the following proposition.

\begin{prop}
There exist constants $C_1,C_2>0$ dependent on the domain $\om$ such that \[C_1\leq\Psi\leq C_2.\]
\label{prop}
\end{prop}
\begin{proof}
For a bounded domain $\om$, let us consider a general elliptic operator
\[L=-\sum\limits_{i,j=1}^2{\pd_{x_i}(a^{ij}\pd_{x_i})}.\]
The coefficients $a^{ij}$ are real-valued from $L^\infty(\om)$ and satisfy
\[\mu_1|\xi|^2\leq\sum\limits_{i,j=1}^2a^{ij}\xi_i^*\xi_j\leq\mu_2|\xi|^2\]
for some $\mu_1,\mu_2>0.$ Let $\mathcal{G}=\mathcal{G}(x,y)$ be the Green's function for $L$ with the homogeneous boundary condition on $\pd\om.$ Then, $\mathcal{G}>0$ and for $x,y\in\om$,
\begin{equation}
\mathcal{G}(x,y)\leq C_1(|\mathrm{ln}|x-y||+1)
\label{C1}
\end{equation}
and
\begin{equation}
\mathcal{G}(x,y)\geq C_2(|\mathrm{ln}|x-y||+1)
\label{C2}
\end{equation}
for $|x-y|\leq\frac{1}{2}\mathrm{dist}(y,\pd\om).$ By results in \cite{LiStWe}, it is sufficient to verify this for the Laplacian for which it is known. Here, $C_1$ and $C_2$ are positive constants that depend only on $\mu_1$ and $\mu_2$. The function $\Psi$ is represented as \[\Psi(x)=2\int\limits_\om{\mathcal{G}(x,y)\dif y,}\] where $\mathcal{G}$ now represents the Green's function for the operator $\beta^{-1}\nad\cdot\beta\nad.$ The required estimate follows from \eqref{C1} and \eqref{C2}.
\end{proof}

We define the generalized torsional rigidity
\begin{equation*}
G(s):=2\!\!\int\limits_{\om(s)}\!\!{\Psi(\eta,\theta,s)\dif\sigma(\eta,\theta)}=\int\limits_{\om(s)}\!\!{\beta(\eta,\theta,s)|\nad\Psi(\eta,\theta,s)|^2\dif\sigma(\eta,\theta)}>0.
%\label{55}
\end{equation*}
Due to this definition and the boundedness of the domain $\omega(s)$,  Proposition \ref{prop} guarantees the existence of constants $A,B$ such that $0<A\leq G(s)\leq B<\infty$ for all $s\in[0,1].$

Now we consider the next step in the asymptotic procedure, that is to compare the coefficients of the next order terms. We have
\begin{equation}
\begin{aligned}
	-\beta^{-1}\nad\cdot\beta\nad\bv^2_\ddagger+\nad p^2=\beta^{-2}h\bc'''_\ddagger v^1_3+\beta^{-3}h\bc''(2\beta\pd_s v^1_3-v^1_3\pd_s\beta),\\\quad-\nad\cdot\beta\bv^2_\ddagger=\pd_sv^1_3\quad\mbox{in}\quad\om(s),\quad
\bv^2_\ddagger=\bnt_\ddagger\quad\mbox{on}\quad\pd\om(s).
\end{aligned}
\label{56}
\end{equation}

Owing to the zero boundary conditions for $\bv^2_\ddagger$ and $\Psi$, we get the compatibility condition for this problem
\begin{align*}
	0&=\int\limits_{\om(s)}\!\!{\nad\cdot\beta\bv^2_\ddagger\dif\sigma(\eta,\theta)}=-\int\limits_{\om(s)}\!\!{\pd_sv^1_3\dif\sigma(\eta,\theta)}\\&=\frac{1}{2}\!\!\int\limits_{\om(s)}\!\!{\pd_s(\Psi(\eta,\theta,s)\pd_sp^0(s))\dif\sigma(\eta,\theta)}=\frac{1}{2}\pd_s(\!\!\int\limits_{\om(s)}\!\!{\Psi(\eta,\theta,s)\dif\sigma(\eta,\theta)}\pd_sp^0(s)).
\end{align*}
Thus, we have derived the modified Reynolds equation
\begin{equation}
-\pd_s(G(s)\pd_sp^0(s))=0,\quad s\in(0,1).
\label{Rey}
\end{equation}
\begin{Rem}
In the absence of curvature, i.e. $\beta\equiv1$, \eqref{Rey} is the classical Reynolds equation, cf. \cite{NaPi90}.
\end{Rem}
which motivates the imposition of the boundary flux condition
%along with the boundary condition
\begin{align}\nonumber
\int\limits_{\om(0)}\!\!v_3^h(r,\theta,0)\dif\sigma(\eta,\theta)&=h^{-1}F^0\\\label{BCR}
\Rightarrow-G(0)\pd_sp^0(0)&=4F^0.
\end{align}
We can also argue to impose the condition
\begin{equation}
p^0(1)=p^0_{per}.
\label{per}
\end{equation}
%where $p^h_{per}=h^{-3}p^0_{per}.$

The mixed boundary problem stated in \eqref{Rey}, \eqref{BCR} and \eqref{per} has the solution \[p^0(s)=p^0_{per}+4F^0\int\limits_s^1{G(t)^{-1}\dif t}.\]
This leads to
\[|\pd_sp^0(s)|=|4F^0G(s)^{-1}|\leq c\]
as well as
\[|\pd_s^2p^0(s)|=|4F^0G(s)^{-2}\pd_sG(s)|\leq c\|\pd_s\Psi\|_{L^2(\om(s))}\leq c(\lambda+\gamma)\]
where we used \eqref{dpsi}. Similarly,
\[|\pd_s^3p^0(s)|\leq c(\|(\pd_sR)\pd_s\Psi\|_{L^2(\pd\om(s))}+\|\pd_s^2\Psi\|_{L^2(\om(s))}\leq c(\lambda^*+\gamma^*+h^{-1/2}\lambda^{3/2}+\gamma^2)\]
where we have used \eqref{ddpsi} and the fact that $\|\pd_s\Psi\|_{L^2(\pd\om(s))}\leq c(\|\pd_s\Psi\|_{H^1(\pd\om(s))})$.

As a result of the above along with \eqref{psi}, \eqref{dpsi} and \eqref{ddpsi}, \eqref{Pois} gives us
\begin{equation}
\begin{aligned}
	\|v^1_3\|\leq ch,\quad\quad\|\pd_sv^1_3\|\leq ch(\lambda+\gamma)\\\quad\mbox{and}\quad\|\pd_s^2v^1_3\|\leq ch(\lambda^*+\gamma^*+h^{-1/2}\lambda^{3/2}+\gamma^2).
\end{aligned}
\label{bnds:v3,p0}
\end{equation}

Here and henceforth, $\|\cdot\|$ shall denote the usual norm on $L^2(\Omh)$ and all constants $c$ will be of the form $c=C(|F^0|+|p^0_{per}|)$ where $C$ is independent of $F^0$, $p^0_{per}$.

Consequently, by \eqref{mod1a}, since $\pd_sv^1_3\in L^2(\om(s))$ and $v^1_3\in L^2(\om(s))\subset H^{-1}(\om(s)),$ denoting the average of $p^2$ over the cross-section by $\bar{p}^2$ we have for the solution of \eqref{56},
\begin{gather}\nonumber
\|\bv^2_\ddagger\|_{H^1(\om(s))}+\|p^2-\bar{p}^2\|_{L^2(\om(s))}\leq c(\lambda\|v^1_3\|_{H^{-1}(\om(s))}+\|\pd_sv^1_3\|_{L^2(\om(s))})\\\label{bnds:v2,p2}
\Rightarrow\|\nad\bv^2_\ddagger\|+\|\bv^2_\ddagger\|+\|p^2-\bar{p}^2\|\leq ch(\lambda+\gamma)
\end{gather}
and similarly by \eqref{mod1c}
\begin{gather}\nonumber
\|(\pd_s\bv^2_\ddagger)_\ddagger\|_{H^1(\om(s))}+\|\pd_s(p^2-\bar{p}^2)\|_{L^2(\om(s))}\leq ch(\lambda^*+\gamma^*+h^{-1/2}\lambda^{3/2}+\gamma^2)\\\label{bnds:dv2}
\Rightarrow\|\pd_s\bv^2_\ddagger\|=\|(\pd_s\bv^2_\ddagger)_\ddagger-\bc''\cdot\bv^2_\ddagger\bc'\|\leq ch(\lambda^*+\gamma^*+h^{-1}\lambda+\gamma^2).
\end{gather}

\section{Boundary conditions at the ends}\label{bndry}

In order to solve the Stokes problem, we need to specify appropriate boundary conditions at the inlet and the outlet. We consider the domain $\Omh$ to be an arbitrarily chosen segment of a much larger pipe in which the fluid is injected at one end and it flows out at the other. Such conditions at the end cross-sections are extremely difficult to model reasonably hence we restrict ourselves to the chosen segment, possibly far away from the ends. Imposing artificial boundary conditions at the ends of the chosen segment gives rise to the boundary layer phenomena near those ends. It brings about a quick variability near the end cross-sections in the solution $\{\bv^h,p^h\}$ of the problem. Although, from a practical point of view, it is absurd to expect such quick variability at arbitrarily chosen portions of the full pipe.

The function of the boundary layer terms in the solutions is to reduce the discrepancy in the artificial boundary conditions. We want to impose such boundary conditions which make the discrepancy as small as possible. One can of course formulate elaborate sets of conditions to achieve this. We however, opt for the simpler way of preparing the boundary data in accordance to our approximations. We take the traces of our approximate fields at the end cross-sections and use them as the boundary data. Thus we reduce the discrepancy at the boundaries to zero while also diminishing the error estimates.

\subsection{Boundary conditions on the cross-section $\omega^h(0)$.}%\label{S1}
We note that the components of the Ans\"atze \eqref{Ans} at the point $s=0$ are completely determined by the data of the problem. Indeed, according to the boundary condition \eqref{BCR} we have
\[
v^1_3(\eta,\theta,0)=-\frac{1}{2}\Psi(\eta,\theta,0)\pd_sp^0(0)=2G(0)^{-1}F^0\Psi(\eta,\theta,0).
\]

In this case, the right-hand side of the continuity equation,
\[
\pd_s\Psi(\eta,\theta,0)\pd_sp^0(0)+\Psi(\eta,\theta,0)\pd_s^2p^0(0),
\]
can be evaluated by using the boundary condition \eqref{BCR} and the differential equation \eqref{Rey}.

Thus, we should take the boundary conditions
\begin{align}\label{37}
v^h_3(r,\theta,0)&=2h^{-1}G(0)^{-1}F^0\Psi(\eta,\theta,0),\\\label{38}
\bv^h_\ddagger(r,\theta,0)&=\bnt_\ddagger\quad\mbox{on}\quad\om^h(0).
\end{align}

\subsection{Boundary conditions on the cross-section $\omega^h(1)$.}%\label{S1}

Since the fluxes through the cross-sections do not change, the expressions (so called velocities of pseudo-deformations) generated by the ansatz \eqref{Ans},
\begin{equation}\label{34}
h^{-1}\beta^{-1}\pd_s v^1_3(\eta,\theta,1)-h^{-3}p^0(1),\quad h^{-1}\beta^{-1}\pd_sv^2_j(\eta,\theta,1),j=1,2,
\end{equation}
can be also evaluated by using the problem's data. The term $h^{-3}p^0(1)$ is essentially ($h^{-2}$ times) larger than the other in \eqref{34}. Therefore, we can take
\begin{equation}
\beta^{-1}\pd_s v^h_3(r,\theta,1)-p^h(r,\theta,1)=-p^h_{per}\quad\mbox{on}\quad\om^h(1)
\label{35}
\end{equation}as one of the boundary conditions on the cross-section $\omega^h(1)$ with $p^h_{per}=h^{-3}p^0_{per}.$ We emphasize that the pressure itself can be taken in the boundary condition since it does not appear in the Green formula for the Stokes system alone. Further, we complement \eqref{35} by the following conditions:
\begin{equation}
v^h_\ddagger(r,\theta,1)=\bnt_\ddagger\quad\mbox{on}\quad\om^h(1).
\label{36}
\end{equation}
Due to the continuity equation \eqref{Div} and relation \eqref{36}, we have
\[
0=\na_\bx\cdot v^h(r,\theta,1)=\beta^{-1}\pd_sv^h_3(r,\theta,1)
\]
and hence the boundary conditions \eqref{35} and \eqref{36} lead to a constant pressure on the cross-section $\omega^h(1)$.

\section{Estimates of the asymptotical remainder terms in the Stokes problem.}

Recall that the solution $(\bv^h,p^h)$ of the problem \eqref{Stokes}-\eqref{NoSlip} is represented as
\begin{equation*}
\bv^h=\bbv+\vrem\quad\mbox{and}\quad p^h=\bbp+\prem
%\label{Sol}
\end{equation*}
where we take the approximate solution to be
\begin{equation*}%\label{71}
\begin{aligned}
					\bbp(\bx)&=h^{-3}p^0(s),\\
\bbv(\bx)&=h^{-1}v^1_3(\eta,\theta,s)\bc'(s)+X^h(s)\bv^2_\ddagger(\eta,\theta,s),\\
 %\bbv_\ddagger(\bx)&=X^h(s)\bv^2_\ddagger(\eta,\theta,s),
\end{aligned}
\end{equation*}
where  $X^h\in C^\infty(0,1)$ is a cut-off function such that $0\leq X^h\leq 1$, $|\pd^p_sX^h(s)|\leq ch^{-p}$,
\begin{equation*}
X^h(s)=
\begin{cases}
1 & \mbox{ when }s\in(2h,1-2h),\\
0 & \mbox{ when }s\in(0,h)\cup(1-h,1)
\end{cases}
%\label{72}
\end{equation*}
and $v^1_3$ and $\bv^2_\ddagger$ are solutions of \eqref{55} and \eqref{56} respectively. Inclusion of the term $\bv^2_\ddagger$ makes the approximate velocity divergence-free inside the channel away from the ends. Moreover, the order of magnitude of $|\bv^2_\ddagger|$ grows closer to that of the leading term $h^{-1}v^1_3$ with increasing $\lambda$ as is evident from \eqref{bnds:v3,p0} and \eqref{bnds:v2,p2}. Due to the restrictions $v^1_3=0,$ $\bv^2_\ddagger=0$ on $\partial\omega(z)$, the boundary condition \eqref{NoSlip} is met. Introducing the cut-off function ensures that the conditions \eqref{38} and \eqref{36} are fulfilled. The same is true for the condition \eqref{37} due to \eqref{Pois} and \eqref{BCR}.
\begin{Rem}
The term $X^h(s)\bv^2_\ddagger$ in the region $\supp(1-X^h)$ plays the role of a zero-order approximation of the boundary layer near the ends of the pipe.
\end{Rem}

In order to derive estimates of the error terms, we require an approximate velocity that is divergence-free in the entire domain including near the ends. Hence, to compensate for the error in the divergence of $\bbv$ near the inlet and the outlet, we consider $\bw$ which satisfies
\begin{equation}\label{probw}
\begin{aligned}
-\na_\bx\cdot\bw&=h^{-1}\beta^{-1}(1-X^h)\pd_sv^1_3&\mbox{ in }\Omh,\\
\bw&=\bnt\quad&\mbox{ on }\pd\Omh.
\end{aligned}
%\label{76}
\end{equation}
The compatibility condition for this problem is satisfied as
\[\int\limits_{\Omh}{h^{-1}\beta^{-1}(1-X^h)\pd_sv^1_3}\dif\bx=\int\limits_0^1{(1-X^h)h\int\limits_{\om(s)}\pd_sv^1_3\dif\sigma(\eta,\theta)\dif s}=0.\]
Therefore one could apply Lemma \ref{Lem1} to get the corresponding estimates for the vector field $\bw$. However, the estimates can be improved upon by observing that the right hand side vanishes in most of the domain. As $\supp(1-X^h)\subseteq[0,2h]\cup[1-2h,1]$, it suffices to solve the above problem \eqref{probw} in the region $\{\bx\in\Omh:s\in(0,2h)\cup(1-2h,1)\}$ with $\bw$ vanishing on the boundary and then to extend it to the rest of $\Omh$ by setting $\bw=\bnt$ for $s\in[2h,1-2h]$. Note that the compatibility condition \eqref{perpend} is satisfied at both ends. Since this new domain where \eqref{probw} needs to be solved, has comparable size in all directions (order $h$), we may use \eqref{lem2} to conclude
\begin{equation}
\begin{aligned}
\|\nab\bw_\ddagger\|+\|(\pd_s\bw_\ddagger)_\ddagger\|+h^{-1}\|\bw_\ddagger\|+\|\nab w_3\|+\|\pd_sw_3\|
+h^{-1}\|w_3\|\\\leq c\|h^{-1}\beta^{-1}(1-X^h)\pd_sv^1_3\|\leq ch^{1/2}(\lambda+\gamma).
\end{aligned}
\label{78}
\end{equation}

%Since the function \eqref{72} does not depend on the transversal variables $(r,\theta)$, by using \eqref{bnds:v3,p0}, we find:
%\begin{gather}\nonumber
%\bbf(\bx)=-\na_{\!\bx}\cdot\bbv(\bx)=-h^{-1}\beta^{-1}\pd_sv^1_3\\\label{76}
%\Rightarrow\|\bbf\|\leq c(\lambda+\gamma).
%\end{gather}

Let us denote the inner product in $L^2(\Omh)$ by $(\cdot,\cdot)$. The discrepancy in \eqref{Stokes} is 
\begin{gather}\nonumber
\Fhb:=\D_\bx(\bv^h-\bbv)-\na_\bx(p^h-\bbp)=-\D_\bx\bbv+\na_\bx\bbp\\\nonumber
=-(h^{-2}\beta^{-1}\nad\cdot\beta\nad\!+\!\beta^{-1}\pd_s\beta^{-1}\pd_s)(X^h\bv^2_\ddagger\!+\!h^{-1}\bc'v^1_3)\!+\!(h^{-1}\nad\!+\!\beta^{-1}\bc'\pd_s)(h^{-3}p^0).
\end{gather}
Rearranging the terms with respect to orders of $h$, we have
\begin{gather*}
\Fhb=-\beta^{-1}\pd_s\beta^{-1}\pd_s(X^h\bv^2_\ddagger)-h^{-1}\beta^{-1}\pd_s\beta^{-1}\pd_s(\bc'v^1_3)\\
-h^{-2}X^h\beta^{-1}\nad\cdot\beta\nad\bv^2_\ddagger+h^{-3}\bc'(-\beta^{-1}\nad\cdot\beta\nad v^1_3+\beta^{-1}\bc'\pd_sp^0).
\end{gather*}
Applying \eqref{55} to the above, we obtain
\begin{gather*}
\Fhb=-h^{-2}X^h\beta^{-1}\nad\cdot\beta\nad\bv^2_\ddagger-h^{-1}\beta^{-1}\pd_s\beta^{-1}\pd_s(\bc'v^1_3)-\beta^{-1}\pd_s\beta^{-1}\pd_s(X^h\bv^2_\ddagger).
\end{gather*}
Now, let us consider the differences $\tv=\bv^h-\bbv-\bw$ and $\prem=p^h-\bbp$ between the true and approximate solutions. The vector $\tv$ is solenoidal by construction. Then, integration by parts and \eqref{35} give us
%\begin{equation}
\begin{gather*}
(\na_\bx\tv+\na_\bx\bw,\na_\bx\tv)=\|\na_\bx\tv\|^2+(\na_\bx\bw,\na_\bx\tv)\\=\!\int\limits_{\om^h(1)}\!\left(\beta^{-1}\pd_s(v^h_3\!-\!\mathbb{v}^h_3)\!-\!p^h\!+\!\bbp\right)\tvt\dif\sigma(r,\theta)\!-\!(\Fhb,\tv)\\
=-\int\limits_{\om^h(1)}h^{-1}\beta^{-1}(\pd_sv^1_3)\tvt\dif\sigma(r,\theta)-(\Fhb,\tv).
\end{gather*}
Substituting the expression for $\Fhb$, the equation above can be written as
\begin{gather*}
\|\na_\bx\tv\|^2=-(\na_\bx\bw,\na_\bx\tv)+(h^{-2}X^h\beta^{-1}\nad\cdot\beta\nad\bv^2_\ddagger,\tv)\\
-\int\limits_{\om^h(1)}h^{-1}\beta^{-1}(\pd_sv^1_3)\tvt\dif\sigma(r,\theta)+(\beta^{-1}\pd_s\beta^{-1}\pd_s(h^{-1}\bc'v^1_3+X^h\bv^2_\ddagger),\tv)
\end{gather*}
Integrating by parts again and using the fact that $\tv$ is divergence-free, we get
\begin{equation}
\begin{aligned}
\|\na_\bx\tv\|^2=-h^{-1}(\beta^{-1}v^1_3\bc'',\beta^{-1}\pd_s\tv)-h^{-1}(\beta^{-1}\bc'\pd_s(v^1_3),\beta^{-1}\pd_s\tv)\\
													-h^{-1}(X^h\nad\bv^2_\ddagger,\nab\tv)-(\beta^{-1}\pd_s(X^h\bv^2_\ddagger),\beta^{-1}\pd_s\tv)-(\na_\bx\bw,\na_\bx\tv).
\end{aligned}
\label{79}
\end{equation}

Let us now estimate the terms in \eqref{79} by using the previously derived estimates. At this point, we shall assume that the parameters $\lambda$, $\gamma$, $\lambda^*$ and $\gamma^*$ are such that \[\gamma=O(\lambda),\quad\gamma^*=(\lambda^*),\quad\mbox{and}\quad\lambda^*=O(h^{-1/2}\lambda^{3/2}).\]

With the help of \eqref{bnds:v3,p0} and \eqref{2der}, the first term in the right hand side of \eqref{79} can be estimated as
\[h^{-1}|(\beta^{-1}v^1_3\bc'',\beta^{-1}\pd_s\tv)|\leq ch^{-1}hh^{-1/2}\lambda^{1/2}	 \|\pd_s\tv\|\leq ch^{-1/2}\lambda^{1/2}\|\na_\bx\tv\|.\]
Using \eqref{bnds:v3,p0} to estimate the second term, we have
\[h^{-1}|(\beta^{-1}\bc'\pd_s(v^1_3),\beta^{-1}\pd_s\tv)|\leq ch^{-1}h\lambda\|\pd_s\tv\|\leq ch^{0}\lambda\|\na_\bx\tv\|.\]
We use \eqref{bnds:v2,p2} to get
\[h^{-1}|(X^h\nad\bv^2_\ddagger,\nab\tv)|\leq ch^{-1}h\lambda\|\nab\tv\|\leq ch^{0}\lambda\|\na_\bx\tv\|.\]
Then due to \eqref{bnds:dv2} and \eqref{2der}, we have
\[|(\beta^{-1}\pd_s(X^h\bv^2_\ddagger),\beta^{-1}\pd_s\tv)|\leq ch^{1/2}\lambda^{3/2}\|\pd_s\tv\|\leq ch^{1/2}\lambda^{3/2}\|\na_\bx\tv\|.\]
Finally, \eqref{78} gives us
\[|(\na_\bx\bw,\na_\bx\tilde{\bv}^h)|\leq ch^{1/2}\lambda\|\na_\bx\tv\|.\]
Also, by Friedrichs's inequality,
\[
\|\tv\|\leq ch\|\na_\bx\tv\|
\]
for the curved cylinder $\Omh$.
Thus, for the discrepancy in the velocity, we arrive at the estimate 
\begin{equation}
\|\na_\bx\tv\|+h^{-1}\|\tv\|\leq ch^{-1/2}\lambda^{1/2}.
\label{Disv}
\end{equation}
Here we note that $\|\na_\bx\bw\|$ is $O(h^{1/2}\lambda)$ while $\|\na_\bx\tv\|$ is $O(h^{-1/2}\lambda^{1/2})$ which leads to \[\|\na_\bx\vrem\|\leq\|\na_\bx\tv\|+\|\na_\bx\bw\|\leq ch^{-1/2}\lambda^{1/2}.\] Moreover, as $\|\bw\|$ is $O(h^{3/2}\lambda)$ and $\|\tv\|$ is $O(h^{1/2}\lambda^{1/2})$, \[\|\vrem\|\leq\|\tv\|+\|\bw\|\leq ch^{1/2}\lambda^{1/2}.\] On the other hand, $\|\na_\bx\bbv\|$ is $O(h^{-1})$ whereas $\|\bbv\|$ is $O(h^0)$. Thus it is safe to conclude that the approximation of velocity is justified for a $\lambda$ which is $O(h^{-1+2\delta})$ for any $\delta>0.$

Let us now estimate the discrepancy in the approximation of pressure. Let us denote the average of a scalar field over $\Omh$ by placing a bar over the corresponding symbol. Consider the velocity field $\tw$ such that
\begin{equation}\nonumber
\begin{aligned}
-\na_\bx\cdot\tw&=\prem-\premb&\mbox{ in }\Omh,\\
						 \tw&=\bnt	 \quad&\mbox{ on }\pd\Omh.
\end{aligned}
%\label{76}
\end{equation}
Clearly, the compatibility condition \eqref{perp} is satisfied.

Then, integration by parts and \eqref{35} result in
\begin{gather*}
(\na_\bx\tv+\na_\bx\bw,\na_\bx\tw)+\|\prem-\premb\|^2\\=\!\!\!\!\int\limits_{\om^h(1)}\!\!\!\!\left(\beta^{-1}\pd_s(v^h_3\!\!-\mathbb{v}^h_3+w_3)-p^h\!\!+\bbp\right)\tilde{w}_3\dif\sigma(r,\theta)-(\Fhb,\tw)\\
=-\int\limits_{\om^h(1)}h^{-1}(\beta^{-1}\pd_sv^1_3)\tilde{w}_3\dif\sigma(r,\theta)-(\Fhb,\tw).
\end{gather*}
Similar steps as before lead us to
\begin{equation}
\begin{aligned}
\|\prem-\premb\|^2&=-(\na_\bx\tv,\na_\bx\tw)-(\na_\bx\bw,\na_\bx\tw)\\
&-h^{-1}(\beta^{-1}v^1_3\bc'',\beta^{-1}\pd_s\tw)-h^{-1}(\beta^{-1}\bc'\pd_s(v^1_3),\beta^{-1}\pd_s\tw)\\
&-h^{-1}(X^h\nad\bv^2_\ddagger,\nab\tw)-(\beta^{-1}\pd_s(X^h\bv^2_\ddagger),\beta^{-1}\pd_s\tw).
\end{aligned}
\label{80}
\end{equation}
Using \eqref{Disv} and Lemma \ref{Lem1}, we estimate the first term on the right hand side of \eqref{80} as
\begin{align*}
											|(\na_\bx\tv,\na_\bx\tw)|&\leq ch^{-1/2}\lambda^{1/2}\|\na_\bx\tw\|\leq ch^{-3/2}\lambda^{1/2}\|\prem-\premb\|.
\end{align*}
Due to \eqref{78}, for the next term, we have
\begin{align*}
											|(\na_\bx\bw,\na_\bx\tw)|&\leq ch^{1/2}\lambda\|\na_\bx\tw\|\leq ch^{-1/2}\lambda^{1/2}\|\prem-\premb\|.
\end{align*}
Once again by Lemma \ref{Lem1}, \eqref{bnds:v3,p0} and \eqref{2der}, we get
\begin{align*}
				|h^{-1}(\beta^{-1}v^1_3\bc'',\beta^{-1}\pd_s\tw)|\leq ch^{-1}hh^{-1/2}\lambda^{1/2}\|\na_\bx\tw\|&\leq ch^{-3/2}\lambda^{1/2}\|\prem-\premb\|,\\
	|h^{-1}(\beta^{-1}\bc'\pd_s(v^1_3),\beta^{-1}\pd_s\tw)|\leq ch^{-1}h\lambda\|\na_\bx\tw\|	 &\leq ch^{-1}\lambda\|\prem-\premb\|.
\end{align*}
We use \eqref{bnds:v2,p2} to get
\[h^{-1}|(X^h\nad\bv^2_\ddagger,\nab\tw)|\leq ch^{-1}h\lambda\|\na_\bx\tw\|\leq ch^{-1}\lambda\|\prem-\premb\|.\]
Then due to \eqref{bnds:dv2} and \eqref{2der}, we have
\[|(\beta^{-1}\pd_s(X^h\bv^2_\ddagger),\beta^{-1}\pd_s\tw)|\leq ch^{1/2}\lambda^{3/2}\|\na_\bx\tw\|\leq ch^{-1/2}\lambda^{3/2}\|\prem-\premb\|.\]

Thus, for the discrepancy in the pressure, we get
\[
\|\prem-\premb\|\leq ch^{-3/2}\lambda^{1/2}.
\]
As $\|\bbp-\overline{\mathbb{p}}^h\|$ is $O(h^{-2})$, once again we see that the approximate pressure upto a constant is justified for a $\lambda$ which is $O(h^{-1+2\delta})$ for any $\delta>0.$

To summarize, we have shown that \eqref{rest} and \eqref{appest} hold under the assumptions \eqref{asmp} thereby justifying our asymptotic approximations.

\section{The case of $O(1)$ curvature}

The more conventional method to tackle the problem in the case of a mildly curving pipe, where we assume that $\bc$ is a smooth function whose derivatives are bounded independently of $h$, would be to expand the scale factor $\beta$ as $1\!-\!h\eta\bc''\!\!\cdot\!\be_1$, instead of keeping it as a parameter for the asymptotic procedure. Let us compare the results obtained by our method in this case with those obtained with the conventional method as mentioned. For this case, the assumptions on the geometry of the centre curve are such that
\[|\bc'''|\leq ch^0,\quad|\bc''''|\leq ch^0,\quad|\pd_sR|\leq ch^0\quad\mbox{and}\quad|\pd_s^2R|\leq ch^0.\]
The first inequality above implies that the curvature has the restriction \[|\bc''|\leq ch^0.\]With these assumptions, let the solution $\{\bv^h,p^h\}$ admit the following formal asymptotic expansions due to the conventional method:
\begin{align*}
	p^h=h^{-3}q^0+h^{-2}q^1+\cdots,\\
	\bv^h=h^{-1}\mathbf{u}^1+\mathbf{u}^2+\cdots.
\end{align*}
%It can be shown that $q^0$ is given by the traditional Reynolds equation (without dependence on the geometry) and subsequently that
%\[q^0(s)=q^0_{per}+4F^0\int\limits_s^1{G_0(t)^{-1}\dif t}\mbox{ and }q^1(s)= 4F^0\int\limits_s^1{G(t)^{-1}\dif t}.\]

Then, \eqref{Stokes}, \eqref{Div} and \eqref{NoSlip} imply that $u^1_3$ and $q^0=q^0(s)$ satisty
\begin{equation}\label{u1q0}
\begin{aligned}
-\Dd u^1_3+\pd_sq^0=0\quad&\mbox{in}\quad\om(s),\\
u^1_3=0\quad\mbox{on}\quad\pd\om(s).
\end{aligned}
\end{equation}
where as $u^2_3$ and $q^1$ fulfill
\begin{equation}\label{u2q1}
\begin{aligned}
-\Dd u^2_3+\pd_sq^1=-\bc''\cdot\nad u^1_3-\eta\bc''\cdot\be_1\pd_sq^0\quad&\mbox{in}\quad\om(s),\\
u^2_3=0\quad\mbox{on}\quad\pd\om(s).
\end{aligned}
\end{equation}
It can be shown that $\mathbf{u}_\ddagger^1$ is still $\bnt_\ddagger$ as well as $q^1=q^1(s)$. One can obtain $\mathbf{u}_\ddagger^2$ and $q^2$ from
\begin{equation}\label{u2q2}
\begin{aligned}
	-\Dd\mathbf{u}_\ddagger^2+\nad q^2=\bnt_\ddagger,\quad-\nad\cdot\mathbf{u}^2_\ddagger=\pd_su^1_3\quad\mbox{in}\quad\om(s),\\
\mathbf{u}^2_\ddagger=\bnt_\ddagger\quad\mbox{on}\quad\pd\om(s)
\end{aligned}
\end{equation}
whereas, for the next terms, we have
\begin{equation}\label{u3q3}
\begin{aligned}
	-\Dd\mathbf{u}_\ddagger^3+\nad q^3=\bc'''_\ddagger u^1_3+2\bc''\pd_su^1_3-\bc''\cdot\nad\mathbf{u}_\ddagger^2,\\
	-\nad\cdot\mathbf{u}^3_\ddagger=\pd_su^2_3-\bc''\cdot\mathbf{u}_\ddagger^2+\eta\bc''\cdot\be_1\pd_su^1_3\quad\mbox{in}\quad\om(s),\\
\mathbf{u}^3_\ddagger=\bnt_\ddagger\quad\mbox{on}\quad\pd\om(s).
\end{aligned}
\end{equation}

Due to \eqref{u1q0}, we have \[u_3^1=-\frac{1}{2}\Psi_0\pd_sq^0\] where the function $\Psi_0$ is the solution of
\begin{equation}
-\Dd\Psi_0=2\quad\mbox{in}\quad\om(s),\quad\Psi_0=0\quad\mbox{on}\quad\pd\om(s).
\label{Pran0}
\end{equation}
Similarly, \eqref{u2q1} gives \[u_3^2=-\frac{1}{2}(\Psi_1\pd_sq^0+\Psi_0\pd_sq^1)\] where the function $\Psi_2$ is the solution of
\begin{equation}
-\Dd\Psi_1=2\eta\bc''\cdot\be_1-\bc''\cdot\nad\Psi_0\quad\mbox{in}\quad\om(s),\quad\Psi_1=0\quad\mbox{on}\quad\pd\om(s).
\label{Pran1}
\end{equation}
For the discrepancy in approximating the function $\Psi$ obtained by our method with $\Psi_0+h\Psi_1$ we find using \eqref{Pran}, \eqref{Pran0} and \eqref{Pran1} that
\begin{gather*}
%\begin{aligned}
	-\beta^{-1}\nad\!\cdot\!\beta\nad(\Psi-\Psi_0-h\Psi_1)=-\Dd(\Psi-\Psi_0-h\Psi_1)+h\beta^{-1}\bc''\!\cdot\!\nad(\Psi-\Psi_0-h\Psi_1)\\
	=2\beta^{-1}-2-h(2\eta\bc''\cdot\be_1-\bc''\cdot\nad\Psi_0)-h\beta^{-1}\bc''\!\cdot\!\nad\Psi_0-h^2\beta^{-1}\bc''\!\cdot\!\nad\Psi_1\\
	=2(\beta^{-1}-1-h\eta\bc''\cdot\be_1)-h^2\beta^{-1}\bc''\!\cdot\!\nad\Psi_1=O(h^2).
%\end{aligned}
\end{gather*}
Thus we conclude that $\Psi_0+h\Psi_1$ approximates $\Psi$ up to order $h^2$. It follows that, defining
\[G_i(s):=2\!\!\int\limits_{\om(s)}\!\!{\Psi_i(\eta,\theta,s)\dif\sigma(\eta,\theta)},\quad i\in\{0,1\},\]
we get an approximation $G_0+hG_1$ for the function $G$ with error $O(h^2)$.

Note that due to the boundary and the divergence conditions in \eqref{u2q2},
\begin{gather*}
\int\limits_{\om(s)}\!\!{\left(\bc''\cdot\mathbf{u}_\ddagger^2-\eta\bc''\cdot\be_1\pd_su^1_3\right)\dif\sigma(\eta,\theta)}
=\int\limits_{\om(s)}\!\!{\left(\bc''\cdot\mathbf{u}_\ddagger^2+\eta\bc''\cdot\be_1\nad\cdot\mathbf{u}_\ddagger^2\right)\dif\sigma(\eta,\theta)}\\
=\int\limits_{\om(s)}\!\!{\nad\cdot(\eta\bc''\cdot\be_1\mathbf{u}_\ddagger^2)\dif\sigma(\eta,\theta)}=0.
\end{gather*}
Hence, the compatibility conditions in \eqref{u2q2} and \eqref{u3q3} respectively provide the equations for $q^0$ and $q^1$ as
\begin{align*}
	-\pd_s(G_0(s)\pd_sq^0(s))=0\quad\mbox{and}\quad
	-\pd_s(G_0(s)\pd_sq^1(s)+G_1(s)\pd_sq^0(s))=0,\quad s\in(0,1).
\end{align*}
We use the same boundary conditions as in \eqref{BCR} and \eqref{per} so that
\begin{alignat*}{2}
	-G_0(0)\pd_sq^0(0)&=4F^0\quad&&\mbox{and}\quad q^0(1)=p^0_{per},\\
	-G_0(0)\pd_sq^1(0)-G_1(0)\pd_sq^0(0)&=0\quad&&\mbox{and}\quad q^1(1)=0.
\end{alignat*}
Thus we have the solutions
\begin{align*}
	q^0(s)&=p^0_{per}+4F^0\int\limits_s^1{\frac{1}{G_0(t)}\dif t},\\
	q^1(s)&=				 -4F^0\int\limits_s^1{\frac{G_1(t)}{G_0(t)^2}\dif t}.
\end{align*}
Then for the discrepancy in the approximations in pressure, we have
\begin{gather*}
	p^0(s)-q^0(s)-hq^1(s)=4F^0\int\limits_s^1{\left(\frac{1}{G(t)}-\frac{1}{G_0(t)}\left(1-\frac{hG_1(t)}{G_0(t)}\right)\right)\dif t}\\
	=4F^0\int\limits_s^1{\left(\frac{1}{G(t)}-\frac{1}{G_0(t)+hG_1(t)}+O(h^2)\right)\dif t}=O(h^2).
\end{gather*}
Now let us consider the difference in the velocity components given by the two methods. For the longitudinal part, we have
\begin{gather*}
	2|v_3^1-u_3^1-hu_3^2|=|\Psi\pd_sp^0-\Psi_0\pd_sq^0-h(\Psi_0\pd_sq^1+\Psi_1\pd_sq^0)|\\
	=|(\Psi_0+h\Psi_1)\pd_s(q^0+hq^1)+O(h^2)-(\Psi_0+h\Psi_1)\pd_sq^0-h\Psi_0\pd_sq^1)|=O(h^2).
\end{gather*}
On the other hand, for the transversal components, we consider \eqref{u2q2}, \eqref{u3q3} and \eqref{56} so that
\begin{gather*}
	-\beta^{-1}\nad\cdot\beta\nad(\bv^2_\ddagger-\mathbf{u}^2_\ddagger-h\mathbf{u}^3_\ddagger)+\nad (p^2-q^2-hq^3)\\
	=h\bc'''_\ddagger(\beta^{-2} v^1_3-u^1_3)+hc''(2\beta^{-2}\pd_s v^1_3-2\pd_s u^1_3-\beta^{-3}v^1_3\pd_s\beta)\\
	=O(h^2)\be_1+O(h^2)\be_2,
\end{gather*}
as well as
\begin{gather*}
	-\nad\cdot\beta(\bv^2_\ddagger-\mathbf{u}^2_\ddagger-h\mathbf{u}^3_\ddagger)=\pd_sv^1_3-\beta(\pd_su^1_3+h\pd_su^2_3)\\
	+h\beta(\bc''\cdot\mathbf{u}_\ddagger^2-\eta\bc''\cdot\be_1\pd_su^1_3)-h\bc''\cdot(\mathbf{u}_\ddagger^2+h\mathbf{u}_\ddagger^2)\\
	=\pd_sv^1_3-\pd_su^1_3-h\pd_su^2_3+h(\beta-1)(\bc''\cdot\mathbf{u}_\ddagger^2-\pd_su^2_3)+(\beta-1)^2\pd_su^1_3=O(h^2).
\end{gather*}

Thus, we conclude that our method produces two-term asymptotic approximations corresponding to a more conventional method for the solution of the problem \eqref{Stokes}, \eqref{Div} and \eqref{NoSlip} in the case of mild curvature.

\appendix

\section{Proof of Claim \ref{Claim}}

\begin{proof}
%\int\limits_{\om(s)}{\beta(\pd_sg+\beta^{-1}\nad\!\cdot\!(\pd_s\beta)\mathbf{u}-\beta^{-2}(\pd_s\beta)\!\cdot\!\beta\nad\mathbf{u})}\dif\sigma(\eta,\theta)\\+\int\limits_{0}^{2\pi}{\beta(\pd_s\mathbf{h}-(\pd_sR)\pd_\eta\mathbf{u})\cdot(R\be_1-(\pd_\theta R)\be_2)\dif\theta}=0.
Firstly, note that the normal $R\be_1-(\pd_\theta R)\be_2=-\pd_\theta(R\be_2)$. Then due to the divergence theorem and \eqref{modeq1a},
\begin{align*}
-\!\int\limits_{\om(s)}\!{\nad\!\cdot\!(\pd_s\beta)\mathbf{u}}\dif\sigma(\eta,\theta)&=\int\limits_{0}^{2\pi}\!{((\pd_s\beta)\mathbf{u})|_{\eta=R}\cdot\pd_\theta(R\be_2)\dif\theta}=\int\limits_{0}^{2\pi}\!{(\pd_s\beta)|_{\eta=R}\mathbf{h}\cdot\pd_\theta(R\be_2)\dif\theta}.
\end{align*}
Once again due to \eqref{modeq1a},
\[-\int\limits_{\om(s)}{\beta^{-1}(\pd_s\beta)\nad\!\cdot\!\beta\mathbf{u}}\dif\sigma(\eta,\theta)=\int\limits_{\om(s)}{(\pd_s\beta)g}\dif\sigma(\eta,\theta).\]
On the other hand, deriving \eqref{compat} with respect to $s$, we get
\begin{align*}
	0&=\int\limits_{\om(s)}{((\pd_s\beta)g+\beta\pd_sg)}\dif\sigma(\eta,\theta)+\int\limits_{0}^{2\pi}{(\pd_sR)(\beta g)|_{\eta=R}R\dif\theta}\\
	 &-\int\limits_{0}^{2\pi}{((\beta\pd_s\mathbf{h}+(\pd_s\beta)\mathbf{h}+(\pd_\eta\beta)(\pd_sR)\mathbf{h})|_{\eta=R}\cdot\pd_\theta(R\be_2)+\beta\mathbf{h}\cdot\pd_s\pd_\theta(R\be_2))|_{\eta=R}\dif\theta}.
\end{align*}
Hence, to prove the claim, it suffices to show that
\[\int\limits_{0}^{2\pi}{((\pd_sR)((\beta\pd_\eta\mathbf{u}+(\pd_\eta\beta)\mathbf{h})\cdot\pd_\theta(R\be_2)-R\beta g)+\beta\mathbf{h}\cdot\pd_s\pd_\theta(R\be_2))|_{\eta=R}\dif\theta}=0.\]
Considering the first term, we have
\[\beta\pd_\eta\mathbf{u}\cdot\pd_\theta(R\be_2)=\beta((\pd_\theta R)\pd_\eta u_2-R\pd_\eta u_1).\]
Then for the next term, due to \eqref{modeq1a} and the fact that $\be_1\cdot\nad=\pd_\eta$, we get
\[(\pd_\eta\beta)\mathbf{h}\cdot\pd_\theta(R\be_2)=\be_1\cdot\nad\beta((\pd_\theta R)u_2-Ru_1).\]
For the third term, we have
\[-R\beta g=R\nad\cdot\beta\mathbf{u}=\beta(R\pd_\eta u_1+u_1+\pd_\theta u_2)+R(u_1\be_1+u_2\be_2)\cdot\nad\beta.\]
Lastly, noting that $\eta\be_2\cdot\nad=\pd_\theta$, we have
\begin{gather*}
	\int\limits_{0}^{2\pi}{\beta|_{\eta=R}\mathbf{h}\cdot\pd_s\pd_\theta(R\be_2)\dif\theta}=-\int\limits_{0}^{2\pi}{\pd_\theta((\beta\mathbf{u})|_{\eta=R})\cdot\pd_s(R\be_2)\dif\theta}\\
	=-\int\limits_{0}^{2\pi}{(\pd_sR)(\beta(u_1+\pd_\theta u_2+(\pd_\theta R)\pd_\eta u_2)+u_2(\be_1\pd_\theta R+\be_2R)\cdot\nad\beta)\dif\theta}.
\end{gather*}
Combining the above, the claim is proved.
\end{proof}

\bibliographystyle{plain}
\bibliography{Bibl4}

\end{document}